\documentclass[atmp]{ipart_v1}

\Vol{25}
\Issue{5}
\Year{2021}
\firstpage{1199}

\usepackage[english]{babel}
\usepackage[utf8]{inputenc}

\usepackage{textcomp}
\usepackage{amstext}
\usepackage{amsthm}
\usepackage{amssymb}
\usepackage{stmaryrd}

\makeatletter

\providecommand{\tabularnewline}{\\}

\numberwithin{equation}{section}
\numberwithin{figure}{section}
\theoremstyle{plain}
\newtheorem{thm}{\protect\theoremname}
  \theoremstyle{plain}
  \newtheorem{lem}[thm]{\protect\lemmaname}
  \theoremstyle{plain}
  \newtheorem{conjecture}[thm]{\protect\conjecturename}

\makeatother

\providecommand{\conjecturename}{Conjecture}
\providecommand{\lemmaname}{Lemma}
\providecommand{\theoremname}{Theorem}

\begin{document}

\title[Heterotic/$F$-theory duality]{Heterotic/$F$-theory duality and Narasimhan-Seshadri equivalence}

\author[H. Clemens and S. Raby]{Herbert Clemens and Stuart Raby}

\begin{abstract}
Finding the $F$-theory dual of a Heterotic model with Wilson-line
symmetry breaking presents the challenge of achieving the dual $\mathbb{Z}_{2}$-action
on the $F$-theory model in such a way that the $\mathbb{Z}_{2}$-quotient
is Calabi-Yau with an Enriques $\mathrm{GUT}$ surface over which
$SU\left(5\right)_{gauge}$ symmetry is maintained. We propose a new
way to approach this problem by taking advantage of a little-noticed
choice in the application of Narasimhan-Seshadri equivalence between
real $E_{8}$-bundles with Yang-Mills connection and their associated
complex holomorphic $E_{8}^{\mathbb{C}}$-bundles, namely the one
given by the real outer automorphism of $E_{8}^{\mathbb{C}}$ by complex
conjugation. The triviality of the restriction on the compact real
form $E_{8}$ allows one to introduce it into the $\mathbb{Z}_{2}$-action,
thereby restoring $E_{8}$- and hence $SU\left(5\right)_{gauge}$
-symmetry on which the Wilson line can be wrapped.
\end{abstract}

\maketitle
\tableofcontents

\section{Introduction}

Duality between Heterotic models and $F$-theory models in String
Theory begins with the compact real form $E_{8}$ of the simple complex
algebraic group $E_{8}^{\mathbb{C}}$. On the Heterotic side one begins
with a Calabi-Yau threefold $V_{3}$ elliptically fibered over a smooth
del Pezzo surface $B_{2}$. $V_{3}$ comes equipped with two bundles
\[
F_{a}\oplus F_{b}
\]
each endowed with a Yang-Mills connection with structure group the
compact real group $E_{8}$. The connection determines and is determined
by its restriction to each elliptic fiber $E_{b_{2}}$ for $b_{2}\in B_{2}$.

$F$-theory begins with an elliptically fibered Calabi-Yau fourfold
$W_{4}$ elliptically fibered over a Fano threefold $B_{3}$ with
origins in the idealized Calabi-Yau fourfold with equation
\begin{equation}
y^{2}=x^{3}+a_{0}z^{5}\label{eq:ideal}
\end{equation}
where $z$ and $a_{0}$ are sections of the anti-canonical bundle
of $B_{3}$. The exceptional fibers of a crepant resolution of $W_{4}$
over points of $S_{\mathrm{GUT}}:=\left\{ z=0\right\} \subseteq B_{3}$
correspond to the positive simple roots of $E_{8}$ intersecting as
dictated by the $E_{8}$-Dynkin diagram, i.e. they map precisely to
the exceptional fibers of the crepant resolution of the $E_{8}$ rational
double point surface singularity
\begin{equation}
\left\{ y^{2}=x^{3}+z^{5}\right\} \subseteq\mathbb{C}^{3}.\label{eq:E8}
\end{equation}
As we will see in the following, this resolution is constructed entirely
within the product of the Lie algebra of the complex algebraic group
$E_{8}^{\mathbb{C}}$ and its set of Borel subalgebras, the latter
being a smooth complex projective manifold \cite{Springer}.

\subsection{Duality in the presence of an order-$2$ element in $\pi_{1}\left(V_{3}^{\vee}\right)$ }

It is often convenient that the Heterotic model, that we will henceforth
denote as $V_{3}^{\vee}$ have an unbranched (Calabi-Yau) double cover
$V_{3}$, or, said otherwise, that a Heterotic model $V_{3}$ admit
a freely acting involution with Calabi-Yau quotient. In particular,
after $E_{8}$-symmetry is broken to $SU\left(5\right)$-symmetry
on $V_{3}/B_{2}$, the last step in symmetry-breaking to what physicists
call the 'Standard Model' MSSM is accomplished by wrapping what is
referred to in string theory literature as a `Wilson line' on the
$\mathbb{Z}_{2}$-quotient $V_{3}^{\vee}/B_{2}^{\vee}$. The $\mathbb{Z}_{2}$
must also act trivially on the Yang-Mills connections on $F_{a}\oplus F_{b}$
so that the Heterotic quotient $V_{3}^{\vee}/B_{2}^{\vee}$ is equipped
with two $E_{8}$-bundles 
\begin{equation}
F_{a}^{\vee}\oplus F_{b}^{\vee},\label{eq:bundles}
\end{equation}
each with an inherited Yang-Mills connection. Said otherwise, the
quotient preserves the $E_{8}$-symmetry of the initial bundles, as
well as its breaking to $SU\left(5\right)$-symmetry on which the
Wilson line is wrapped. 

On the $F$-theory side, the $\mathbb{Z}_{2}$-action may have orbifold
singularities but must be free on the smooth surface $S_{\mathrm{GUT}}\subseteq B_{3}$.
Again the $\mathbb{Z}_{2}$-action must preserve the initial $E_{8}$-symmetry,
as well as the breaking to $SU\left(5\right)$.\footnote{The implications of this last issue seem not to have been fully appreciated
in the literature. Another issue connected with Wilson line breaking
in $F$-theory is the existence of vector-like exotics. We will deal
with that issue separately in a forthcoming paper.}

However, in order that the $F$-theory quotient be Calabi-Yau, any
$\mathbb{Z}_{2}$-action on the $F$-theory side will have to incorporate
the involution
\[
\frac{dx}{y}\mapsto-\frac{dx}{y}
\]
on the relative one-form on the fibers of the elliptically fibered
$F$-theory model $W_{4}/B_{3}$. As is easily seen in (\ref{eq:ideal}),
a consequence is that each $E_{8}$-root, as represented by an exceptional
curve over a point of $S_{\mathrm{GUT}}$, is sent to its negative.
On the other hand, one must start with $E_{8}$-symmetry on the $F$-theory
quotient quotient as well. 

Our conclusion will be that, in fact, whatever the $\mathbb{Z}_{2}$-action
on $W_{4}/B_{3}$ turns out to be, it will have to somehow incorporate
the symmetry
\[
-I_{4}:\mathfrak{h}_{SU\left(5\right)}^{\mathbb{C}}\rightarrow\mathfrak{h}_{SU\left(5\right)}^{\mathbb{C}}
\]
on the complexified Cartan subalgebra of $SU\left(5\right)$, the
involution that sends each root to minus itself, an involution that
is not an element of the Weyl group $W\left(SU\left(5\right)\right)$
but one that does preserve the symmetry with respect to the compact
real group $SU\left(5\right)$.

\subsection{The example of $SU\left(2\right)$}

To illustrate how this can work, we illustrate in the simplest case.
We consider the $A_{1}$ rational double point surface singularity
\[
y^{2}=xz.
\]
This is a quotient singularity via the map from the $\left(u,v\right)$-plane
given by
\[
\begin{array}{c}
y=uv\\
x=u^{2}\\
z=v^{2}.
\end{array}
\]
Assigning $u$ and $v$ weight $1/2$ the versal deformation space
is the weighted homogeneous space
\begin{equation}
y^{2}=xz+a_{2}\label{eq:su2}
\end{equation}
of weight two where the Casimir polynomial algebra on the complexified
Cartan $\mathfrak{h}_{SU\left(2\right)}^{\mathbb{C}}$ has generator
\[
a_{2}:\mathfrak{h}_{SU\left(2\right)}^{\mathbb{C}}\rightarrow\frac{\mathfrak{h}_{SU\left(2\right)}^{\mathbb{C}}}{W\left(SU\left(2\right)\right)}.
\]
The base extension
\[
\left|\left(\begin{array}{cc}
x & y+\sqrt{a_{2}}\\
-y+\sqrt{a_{2}} & z
\end{array}\right)\right|=0
\]
has equivariant crepant resolution parametrized by $\mathfrak{h}_{SU\left(2\right)}^{\mathbb{C}}$
and with exceptional fiber parametrized by either the ratio of the
rows or the ratio of the columns of the above $2\times2$ matrix.
Interchanging rows and columns by transposition is accomplished by
replacing $y$ with $-y$, that is, by acting on the singularity (\ref{eq:su2})
by the automorphism
\[
\left(x,y,z\right)\mapsto\left(x,-y,z\right).
\]
Via the faithful spin representation given by the quaternions, we
may consider $SU\left(2\right)$ as a real matrix group with complexification
$SL\left(2;\mathbb{C}\right)$. What will be relevant for us in this
paper is the following commutative diagram
\[
\begin{array}{ccccc}
 &  & SU\left(2\right)\\
 & \swarrow &  & \searrow\\
SL\left(2;\mathbb{C}\right) &  & \overset{\iota}{\longrightarrow} &  & SL\left(2;\mathbb{C}\right)
\end{array}
\]
where $\iota$ is complex conjugation. The equivariant crepant resolution
of (\ref{eq:su2}) involves a choice of Weyl chamber identifying the
exceptional curve of the resolution with the positive root. The `flop'
$y\mapsto-y$ interchanges the two possible choices of positive simple
root corresponding to the two possible equivariant crepant resolutions.
The necessity of introducing the base extension in order to equivariantly
resolve the family (\ref{eq:su2}) and the `flop' $y\mapsto-y$ interchanging
the two possible equivariant crepant resolutions lies at the heart
of the problem of finding a canonical crepant resolution of the Tate
form (\ref{eq:Tate}) defining an $F$-theory model $W_{4}$. 

The action of $\iota$ on the $\mathfrak{h}_{SU\left(2\right)}^{\mathbb{C}}$
is therefore $h\mapsto-h$. Notice that this action should \textit{not}
be thought of as the action of $W\left(SU\left(2\right)\right)$ since
the action of $\iota$ is trivial on the compact real group $SU\left(2\right)$.
Rather it should be thought of as the transformation that interchanges
each root $\varrho$ with its negative $-\varrho$. In fact for $n>2$,
$\iota$ does not act on roots as an element of the Weyl group, the
difference of the actions being given by the non-trivial symmetry
of the Dynkin diagram.

\subsection{Heterotic $E_{8}$ versus $F$-theory $E_{8}^{\mathbb{C}}$ }

The fact that the Heterotic model relies on the properties of the
real group $E_{8}$ while the $F$-theory model relies on properties
of the complex algebraic group $E_{8}^{\mathbb{C}}$ is a central
theme of this paper. More specifically, as explained below, the passage
from the real to the complex group in constructing Heterotic/$F$-theory
duality rests on the choice between the two possible complexifications
of the same real bundle via Narasimhan-Seshadri equivalence. The $\mathbb{Z}_{2}$-action
must reverse that choice if symmetry of the real group and the Calabi-Yau
property are to be simultaneously preserved on the quotient.

\subsection{Narasimhan-Seshadri equivalence}

Let $G_{\mathbb{R}}$ denote a compact simple real Lie group and let
$G_{\mathbb{C}}$ denote the simple complex algebraic group for which
$G_{\mathbb{R}}$ is the compact real form. Narasimhan-Seshadri equivalence
on a compact Riemann surface $C$ equates homomorphisms
\[
\pi_{1}\left(C\right)\rightarrow G_{\mathbb{R}}
\]
and semi-stable $G_{\mathbb{C}}$ vector bundles on $C$. The relevant
remark is that, since $G_{\mathbb{R}}$ has a faithful real linear
representation, therefore $G_{\mathbb{C}}$ can be defined by extension
of scalars, that is, replacing real matrix entries with complex ones.
Therefore complex conjugation induces a real involution
\[
\iota:G_{\mathbb{C}}\rightarrow G_{\mathbb{C}}
\]
that leaves $G_{\mathbb{R}}$ pointwise fixed. $\iota$ is of course
not complex analytic. Thus the Narasimhan-Seshadri equivalence implicitly
makes a choice of one of the two possible complex structures on the
flat $G_{\mathbb{C}}$-bundle. The purpose of this paper is to point
out a situation in which $G_{\mathbb{C}}=E_{8}^{\mathbb{C}}$ and
the choice matters. This happens when introducing $\mathbb{Z}_{2}$-actions
on $F$-theory/Heterotic dual manifolds with the property that the
respective quotient manifolds continue to be dual.

\subsection{Outline of the paper}

In $F$-theory the exceptional components of the fibers of a crepant
resolution $\tilde{W}_{4}/B_{3}$ of $W_{4}/B_{3}$ are identified
with a system of positive simple roots of $SU\left(5\right)$. What
is often less attended to in the presence of a $\mathbb{Z}_{2}$-action
is the trajectory of those roots as initial $E_{8}$-symmetry is broken.
In particular, on the Heterotic side the initial symmetry on $V_{3}^{\vee}/B_{2}^{\vee}$
is $E_{8}$-symmetry. Therefore a Heterotic dual $\tilde{W}_{4}^{\vee}/B_{3}^{\vee}$
should also manifest initial $E_{8}$-symmetry. 

Section \ref{sec:-theory/Heterotic-Duality2} is devoted to establishing
the fact that, in order that $\tilde{W}_{4}^{\vee}/B_{3}^{\vee}$
be Calabi-Yau, the $\mathbb{Z}_{2}$-action must incorporate the standard
involution
\begin{equation}
\left(x,y\right)\mapsto\left(x,-y\right)\label{eq:00}
\end{equation}
on the Weierstrass form of the elliptic fibers of $W_{4}/B_{3}$.
In addition it is shown how this last is compatible with the fact
that the $\mathbb{Z}_{2}$-action on the Heterotic side that, as it
must, incorporates the trivial involution
\[
\left(x,y\right)\mapsto\left(x,y\right)
\]
on the Weierstrass form of the elliptic fibers of the Heterotic model
$V_{3}/B_{2}$. This Section also reviews the construction of the
semi-stable limit in $F$-theory and the critical role that Narasimhan-Seshadri
equivalence plays there.

Section \ref{sec:Retaining--symmetry3} is devoted to showing that
the one (real) involution on $E_{8}^{\mathbb{C}}$ that leaves $E_{8}$
pointwise fixed, namely complex conjugation, exactly reverses the
sign of each $E_{8}$-root. That is, complex conjugation acts as minus
the identity $\left(-I_{8}\right)$ on the Cartan subalgebra of $E_{8}^{\mathbb{C}}$. 

Section \ref{sec:Device-for-tracking4} employs the Tate form for
$W_{4}/B_{3}$ to imbed it in a family of rational double-point surface
singularities that are in turn mapped into the semi-universal deformation
of the $E_{8}$-rational double-point surface singularity
\[
y^{2}=x^{3}+z^{5}.
\]

Section \ref{sec:Equivariant-crepant-resolution5} examines Brieskorn-Grothendieck
equivariant crepant resolution of the semi-universal deformation of
the $E_{8}$-rational double-point surface singularity. We show that
the involution (\ref{eq:00}) is derived from the central involution
$-I_{8}$ on the complex Cartan subalgebra $\mathfrak{h}_{E_{8}}^{\mathbb{C}}$
so that, by Section \ref{sec:Retaining--symmetry3} it can be built
into the $\mathbb{Z}_{2}$-action without breaking $E_{8}$-symmetry.

Section \ref{sec:Tracking-the-equivariant6} tracks the $SU\left(5\right)$-roots,
as manifest in the exceptional components of a general fiber of a
crepant resolution $\tilde{W}_{4}/B_{3}$ of $W_{4}/B_{3}$ over $S_{\mathrm{GUT}}$,
back to their origins as $E_{8}$-roots exploiting the commutativity
of the three-dimensional commutative diagram obtained by mapping the
top row of 
\[
\begin{array}{ccc}
SL\left(5;\mathbb{C}\right) & \rightarrow & E_{8}^{\mathbb{C}}\\
\uparrow &  & \uparrow\\
SU\left(5\right) & \rightarrow & E_{8}\\
\downarrow &  & \downarrow\\
SL\left(5;\mathbb{C}\right) & \rightarrow & E_{8}^{\mathbb{C}}
\end{array}
\]
to the bottom row by the complex conjugate involution $\iota$. It
is exactly the commutativity of this diagram that allows us to claim
that initial $E_{8}$-symmetry and subsequent $SU\left(5\right)$-symmetry
are preserved on the $F$-theory quotient $W_{4}^{\vee}/B_{3}^{\vee}$.

Finally in Section \ref{sec:Crepant-resolution-conjecture7} we
state a conjecture that, if true, would derive from the Brieskorn-Grothendieck
equivariant crepant resolution and the choice of positive Weyl chamber
the construction of a `canonical' crepant resolution of $W_{4}/B_{3}$. 

\clearpage

\section{$F$-theory/Heterotic Duality\label{sec:-theory/Heterotic-Duality2}}

\subsection{Smooth elliptically-fibered Heterotic theory}

The starting point in the construction of smooth Heterotic theory
is an elliptically-fibered Calabi-Yau threefold $V_{3}/B_{2}$ over
a smooth del Pezzo surface $B_{2}$ such that $V_{3}/B_{2}$ comes
equipped with two two bundles 
\[
F_{a}\oplus F_{b}
\]
with structure group the compact real group $E_{8}$ . Each bundle
is endowed with a Yang-Mills connection, a connection determining
and determined by its restriction to each elliptic fiber $E_{b_{2}}$.
Each restriction is flat and therefore given by a homomorphism
\begin{equation}
\pi_{1}\left(E_{b_{2}}\right)\rightarrow E_{8}.\label{eq:fgprep}
\end{equation}
Since $\pi_{1}\left(E_{b_{2}}\right)$ is abelian, the image of the
homomorphism can be conjugated into a maximal torus of $E_{8}$. So
any semi-stable $E_{8}$-bundle with flat connection on $E_{b_{2}}$
reduces to a unique homomorphism
\begin{equation}
\left\{ \pi_{1}\left(E_{b_{2}}\right)\rightarrow\mathfrak{t}_{E_{8}}=\left(S^{1}\right)^{8}\subseteq\left(\mathbb{C}^{\ast}\right)^{8}\right\} _{b_{2}\in B_{2}}.\label{eq:NS}
\end{equation}

The exact sequence
\[
0\rightarrow\pi_{B_{2}}^{\ast}T_{B_{2}}^{\ast}\rightarrow T_{V_{3}}^{\ast}\rightarrow T_{V_{3}/B_{2}}^{\ast}\rightarrow0
\]
yields the equality

\[
\det T_{V_{3}}^{\ast}=\pi_{B_{2}}^{\ast}\left(\det T_{B_{2}}^{\ast}\right)\otimes K_{V_{3}/B_{2}}.
\]
Since $V_{3}$ is Calabi-Yau, $\det T_{V_{3}}^{\ast}$ is the trivial
line bundle. Thus the $\mathbb{Z}_{2}$-action is either trivial on
both of the right-hand factors or non-trivial on both. Castelnuovo's
Rationality Criterion implies that there are no freely acting involutions
on the del Pezzo surface $B_{2}$. Linearizing the action of $B_{2}$
around fixpoints yields the conclusion that either $\beta_{2}$ acts
with finite fixpoint set and the action on relative one-forms in $K_{V_{3}/B_{2}}$
is
\[
\frac{dx}{y}\mapsto\frac{dx}{y}
\]
or has a fixed curve along which the action on relative one-forms
in $K_{V_{3}/B_{2}}$ is
\[
\frac{dx}{y}\mapsto\frac{-dx}{y}.
\]
We will next see that the existence of an $F$-theory dual implies
that the action of $\beta_{2}$ has only finite fixpoint set.

\subsection{$F$-theory model}

The starting point in the construction of $F$-theory is an elliptically
fibered Calabi-Yau fourfold $W_{4}/B_{3}$ over a smooth Fano threefold
$B_{3}$, itself fibered over the Heterotic $B_{2}$ with rational
fibers. The $F$-theory model must be endowed with equivariant involutions
\[
\begin{array}{ccc}
W_{4} & \overset{\tilde{\beta}_{4}}{\longrightarrow} & W_{4}\\
\downarrow &  & \downarrow\\
B_{3} & \overset{\beta_{3}}{\longrightarrow} & B_{3}\\
\downarrow &  & \downarrow\\
B_{2} & \overset{\beta_{2}}{\longrightarrow} & B_{2}.
\end{array}
\]
Duality then requires that $\beta_{3}$ acts freely on the smooth
anti-canonical divisor $S_{\mathrm{GUT}}\subseteq B_{3}$. Therefore
$S_{\mathrm{GUT}}$ is a $K3$-surface and the quotient under the
free $\mathbb{Z}_{2}$-action is an Enriques surface. Since $S_{\mathrm{GUT}}$
is ample, the involution $\beta_{3}$ can have only finite fixpoint
set, a fact that in turn implies that the Heterotic $\beta_{2}$ can
have only finite fixpoint set. As we have seen above, this implies
that the involution $\tilde{\beta}_{3}$ on the Heterotic $V_{3}$
will have to act as 
\begin{equation}
\frac{dx}{y}\mapsto\frac{dx}{y}\label{eq:inv5}
\end{equation}
on relative one-forms in $K_{V_{3}/B_{2}}$.

The short exact sequence
\begin{equation}
0\rightarrow\pi_{B_{3}}^{\ast}T_{B_{3}}^{\ast}\rightarrow T_{\tilde{W}_{4}}^{\ast}\rightarrow T_{\tilde{W}_{4}/B_{3}}^{\ast}\rightarrow0\label{eq:tbseq}
\end{equation}
of cotangent spaces to a crepant resolution $\tilde{W}_{4}/B_{3}$
of $W_{4}/B_{3}$ yields an equation 
\begin{equation}
\det T_{\tilde{W}_{4}}^{\ast}=\pi_{B_{3}}^{\ast}\left(\det T_{B_{3}}^{\ast}\right)\otimes K_{\tilde{W}_{4}/B_{3}}\label{eq:prod}
\end{equation}
where $\det T_{B_{3}}^{\ast}$ has a meromorphic section $\omega_{\mathrm{GUT}}$
with no zeros and simple pole along the $K3$-surface $S_{\mathrm{GUT}}\subseteq B_{3}$.
Again the $\mathbb{Z}_{2}$-action is either trivial on both of the
right-hand factors or non-trivial on both. The residue of $\omega_{\mathrm{GUT}}$
is a nowhere vanishing holomorphic two-form on $S_{\mathrm{GUT}}$.
Since $\beta_{3}$ must act freely one concludes that

\begin{equation}
\beta_{3}^{\ast}\left(\omega_{\mathrm{GUT}}\right)=-\omega_{\mathrm{GUT}}.\label{eq:GUT}
\end{equation}
So, in order that quotient $W_{4}^{\vee}$ be Calabi-Yau, (\ref{eq:prod})
implies that the involution $\tilde{\beta}_{4}$ must act as 
\begin{equation}
\frac{dx}{y}\mapsto-\frac{dx}{y}\label{eq:inv4}
\end{equation}
on relative one-forms in $K_{\tilde{W}_{4}/B_{3}}$, that is, the
involution $\tilde{\beta}_{4}$ will have to act as
\begin{equation}
\left(x,y\right)\mapsto\left(x,-y\right)\label{eq:inv2}
\end{equation}
on the Weierstrass form on the fibers of $W_{4}/B_{3}$. 

\subsection{Preserving duality of the $\mathbb{Z}_{2}$-quotients from Heterotic
to $F$-theory}

The duality is realized by the canonical replacement of the restriction
of the two bundles 
\[
F_{a}\oplus F_{b}
\]
to each elliptic fiber $E_{b_{2}}$ of $V_{3}/B_{2}$ by a union of
elliptically fibered rational surfaces
\begin{equation}
\left(dP_{a}\left(b_{2}\right)\cup dP_{b}\left(b_{2}\right)\right)\rightarrow\mathbb{P}_{\left[a',a''\right]}\left(b_{2}\right)\cup\mathbb{P}_{\left[b',b''\right]}\left(b_{2}\right)\label{eq:fiber}
\end{equation}
such that
\[
dP_{a}\left(b_{2}\right)\cap dP_{b}\left(b_{2}\right)=E_{b_{2}}.
\]
(\ref{eq:fiber}) is then a normal-crossing $K3$-surface elliptically
fibered over the union of two $\mathbb{P}^{1}$'s meeting at a point.
Taken together these normal-crossing $K3$-surfaces are the fibers
of a fibration sequence
\[
W_{4,0}\rightarrow B_{3,a}\cup B_{3,b}\rightarrow B_{2}
\]
with total space a normal-crossing Calabi-Yau fourfold. 

The above canonical replacement is permitted by three facts:

1) The Yang-Mills connections determine and are determined by their
restrictions to flat connections on each elliptic fiber $E_{b_{2}}$.

2) The Narasimhan-Seshadri theorem allows replacement of the flat
$E_{8}$-bundles 
\[
F_{a}\left(b_{2}\right)\oplus F_{b}\left(b_{2}\right)
\]
on $E_{b_{2}}$ with flat holomorphic $E_{8}^{\mathbb{C}}$-bundles
\[
F_{a}^{\mathbb{C}}\left(b_{2}\right)\oplus F_{b}^{\mathbb{C}}\left(b_{2}\right)
\]
where $E_{8}^{\mathbb{C}}$ is the algebraic group whose compact real
form is $E_{8}$. (This is the point at which one of the two complexifications
of the $E_{8}$-bundles is chosen. As we shall show below, the the
$\mathbb{Z}_{2}$-action must incorporate a reversal of that choice
in order that the $F$-theory quotient retain $E_{8}$-symmetry.)

3) In Section 4.5 of \cite{Friedman} Friedman-Morgan-Witten give
us a classifying space for imbeddings of $E_{b_{2}}$ into a rational
elliptic surface $dP_{9}\left(b_{2}\right)$, each such corresponding
canonically by a theorem of E. Looijenga \cite{Looijenga} to an isomorphism
class of flat $E_{8}^{\mathbb{C}}$-bundles $F$ over $E_{b_{2}}.$

Namely one considers the family of `$dP_{9}$-hypersurfaces' 
\begin{align}
\label{44}
y^{2}&=4x^{3}-\left(g_{2}t^{4}-\beta_{1}st^{3}-\ldots-\beta_{4}s^{4}\right)x\\
\notag &\quad -\left(g_{3}t^{6}-\alpha_{2}s^{2}t^{4}-\ldots-\alpha_{6}s^{6}\right)
\end{align}
in $\mathbb{P}_{1,1,2,3}^{3}$ parametrized by homogeneous forms $\alpha_{j}$
and $\beta_{j}$ of weight $j$ in a weighted projective space $\mathbb{P}_{1,2,2,3,3,4,4,5,6}^{8}$.
Fixing the values of $\alpha_{j}$ and $\beta_{j}$ we think of the
solution set of $\eqref{44}$ as a rational hypersurface in $\mathbb{P}_{1,1,2,3}^{3}$
with distinguished pencil 
\begin{equation}
\gamma s+\delta t=0.\label{444}
\end{equation}
The given elliptic curve $E_{b_{2}}$ sits in each $dP_{9}\left(b_{2}\right)$
in $\eqref{44}$ as the solution set to the equation 
\[
s=0.
\]
The associated sum of eight flat line bundles on $E_{b_{2}}$ is given
by the morphism
\[
H_{0}^{2}\left(dP_{9}\left(b_{2}\right);\mathbb{Z}\right)\rightarrow\mathrm{Pic}^{0}\left(E_{b_{2}}\right)
\]
where $H_{0}^{2}\left(dP_{9}\left(b_{2}\right);\mathbb{Z}\right)$
is the space of algebraic cycles on $dP_{9}\left(b_{2}\right)$ whose
\linebreak intersection number with $E_{b_{2}}$ is zero. The intersection pairing
on\linebreak $H_{0}^{2}\left(dP_{9}\left(b_{2}\right);\mathbb{Z}\right)$ is
that of the $E_{8}$-Dynkin diagram.

Now (\ref{44}) should be thought of as defining a fiber of a bundle
or `stack' over the moduli stack $\mathfrak{M}_{1,1}$ of elliptic
curves given by their Weierstrass form. $\mathfrak{M}_{1,1}$ has
a covering involution
\begin{equation}
\left(\left(x,y\right),\left[s,t\right]\right)\mapsto\left(\left(x,-y\right),\left[-s,t\right]\right)\label{eq:log}
\end{equation}
that lifts to $dP_{9}$-hypersurface involution
\begin{equation}
\begin{array}{c}
\left(\left[s,t\right],\left(x,y\right),\left[\beta_{1},\alpha_{2},\beta_{2},\alpha_{3},\beta_{3},\alpha_{4},\beta_{4},\alpha_{5},\alpha_{6}\right]\right)\\
\shortdownarrow\\
\left(\left[-s,t\right],\left(x,-y\right),\left[-\beta_{1},\alpha_{2},\beta_{2},-\alpha_{3},-\beta_{3},\alpha_{4},\beta_{4},-\alpha_{5},\alpha_{6}\right]\right)
\end{array}\label{eq:parity}
\end{equation}
since the parity of the coefficients $\alpha_{i}$ and $\beta_{j}$
in the weighted projective space $\mathbb{P}_{1,2,2,3,3,4,4,5,6}^{8}$
is matched by their degree. On the other hand (\ref{eq:inv2}) induced on $W_{4,0}/B_{3}$ and
(\ref{eq:inv5}) induced on $V_{3}/B_{2}$ taken together will imply that at $s=0$  a `logarithmic transform' 
(\ref{eq:log}) must be incorporated into the quotienting action along the
elliptic fiber $E_{b_{2}}$. Only in this way can the two components
of the quotient of 
\[
dP_{a}\left(b_{2}\right)\cup dP_{b}\left(b_{2}\right)
\]
by the action of the involution retain the structure of $dP_{9}$'s
without multiple fibers. 

The final step in passing from the Heterotic model to the $F$-theory
model $W_{4}/B_{2}$ is then obtained by smoothing each normal-crossing
$K3$-surface $dP_{a}\left(b_{2}\right)\cup dP_{b}\left(b_{2}\right)$
to obtain a smooth $K3$-surface fibered over the smoothing of $\mathbb{P}_{\left[a',a''\right]}\left(b_{2}\right)\cup\mathbb{P}_{\left[b',b''\right]}\left(b_{2}\right)$
to the fiber of $B_{3}/B_{2}$ over $b_{2}$. Thus one obtains a fibration
sequence
\[
W_{4}\rightarrow B_{3}\rightarrow B_{2}
\]
that consists over generic $b_{2}\in B_{2}$ of a smooth $K3$-surface
elliptically fibered over the $\mathbb{P}^{1}$-fiber of $B_{3}/B_{2}$. 

\subsection{Compatibility of $\mathbb{Z}_{2}$-actions}

The involution $\tilde{\beta}_{4}/\beta_{3}$ on $W_{4}/B_{3}$ with
Calabi-Yau quotient $W_{4}^{\vee}/B_{3}^{\vee}$ must specialize to
an involution $\tilde{\beta}_{4,0}/\beta_{3,0}$ on the semi-stable
limit $W_{4,0}/\left(B_{3,a}\cup B_{3,b}\right)$.

Key to understanding the $\mathbb{Z}_{2}$-action on $W_{4,0}$ is
attending to the $\mathbb{Z}_{2}$-actions on the normal crossing
$K3$-surfaces
\[
\left(dP_{a}\left(b_{2}\right)\cup dP_{b}\left(b_{2}\right)\right)\rightarrow\mathbb{P}_{\left[a',a''\right]}\left(b_{2}\right)\cup\mathbb{P}_{\left[b',b''\right]}\left(b_{2}\right)
\]
over the set $Orb$ of fixpoints $b_{2}$ of the involution $\beta_{2}$.
Since the action of the involution $\tilde{\beta}$ on $V_{3}$ is
free, it must restrict to translation by a non-zero half-period on
the elliptic fiber $E_{b_{2}}$ over the fixed $b_{2}$. 

Over each point of $Orb$, the involution $\tilde{\beta}_{4,0}/\beta_{3,0}$
induces compatible involutions on each of the two components $dP_{9,a}$
and $dP_{9,b}$ of the fiber. The $K3$-surface over a point $b_{2}\in Orb$
induces an involution on each of the two $dP_{9}$-surfaces into which
it splits in the semi-stable limit. The involution must specialize
to translation by a given half-period $\delta$ on the fiber $E_{b_{2}}$,
the intersection of the two $dP_{9}$'s. 
\begin{lem}
\label{lem:i)-On-the F}i) On the $F$-theory fiber of the semi-stable
limit over a fixpoint of $\beta_{2}$, the involution 
\[
\left(\left[s,t\right],\left(x,y\right)\right)\mapsto\left(\left[-s,t\right],\left(x,-y\right)\right)
\]
in (\ref{eq:parity}) acts on each of the two $dP_{9}$'s. Therefore
a section of the canonical bundle of the normal-crossing $K3$-surface
away from its singular locus is given by setting $t=1$ and writing
the holomorphic two form 
\[
ds\wedge\frac{dx}{y}.
\]
Therefore this form is locally invariant under the action of $\tilde{\beta}_{4}$
on $W_{4,0}$ since both factors $ds$ and $\frac{dx}{y}$ are anti-invariant.

ii) On a small analytic neighborhood of $E_{b_{2}}$, the involution is
given on each $dP_{9}$ by the so-called `logarithmic transformation'
\begin{equation}
\left(\left[s,t\right],\left(x,y\right)\right)\mapsto\left(\left[-s,t\right],\left(\left(x,y\right)+\delta\right)\right).\label{eq:parity2}
\end{equation}
However the canonical bundle near the crossing locus of the two $dP_{9}$'s
on the $F$-theory side is represented by two-form
\[
d\log s\wedge\frac{dx}{y}
\]
on each local component\footnote{When smooth surfaces specialize to normal crossing surfaces, a holomorphic
section of their canonical bundle specializes to a meromorphic section
with logarithmic pole with cancelling residues on each of the two
local components.} whose residue is the holomorphic one-form on the Heterotic side.
is invariant under the action of $\tilde{\beta}_{4}$ on $W_{4,0}$
and so the its residual one-form 
\[
\frac{dx}{y}
\]
is invariant under the induced action of $\tilde{\beta}_{3}$ as required
on the Heterotic side.
\end{lem}

\begin{proof}
Since flat bundles on elliptic curves are invariant under translation,
one sees by (\ref{44}) that there are only two possibilities:

1) The $E_{8}$-bundles are pasted to themselves according to the
identity isomorphism induced by translation by the half-period $\delta$.

2) The pasting of each $E_{8}$-bundle incorporates the automorphism
corresponding to the involution
\[
\left(\left[s,t\right],\left(x,y\right)\right)\mapsto\left(\left[-s,t\right],\left(x,-y\right)\right)
\]
 given in (\ref{eq:parity}) on each of the two $dP_{9}$'s.

Possibility 1) is impossible since it would imply that the fiber of
the fibration $B_{3}/B_{2}$ over the fixpoint would be pointwise
invariant under the $\mathbb{Z}_{2}$-action. That would in turn imply
that the $\mathbb{Z}_{2}$-action on the smooth anti-canonical divisor
$S_{\mathrm{GUT}}\subseteq B_{3}$ would also have fixpoints. This
last eliminates the possibility of an $F$-theory quotient.

Possibility 2) however implies that $\mathbb{Z}_{2}$-action on $B_{3}$
has finite fixpoint set thereby allowing a free action on $S_{\mathrm{GUT}}$.
The $\mathbb{Z}_{2}$-action on the two $dP_{9}$-fibers over $b_{2}\in Orb$
is then given on (\ref{44}) by translation by the distinguished half-period
of the common fiber $\left\{ s_{a}=s_{b}=0\right\} $ composed with
(\ref{eq:parity}). For any half-period $\delta$ of $E_{b_{2}}$,
there is defined a so-called `logarithmic transform' on the $dP_{9}$-fiber,
that is, an involution that produces in the quotient a fiber of multiplicity
two over $\left\{ s^{2}=0\right\} $. Setting $t=1,$ the involution
(\ref{eq:parity2}) takes the two-form

\[
ds\wedge\frac{dx}{y}
\]
to minus itself. On the other hand, the involution (\ref{eq:parity})
leaves this same two-form invariant.

However if one removes the multiple fiber ${E_{b_{2}}/{\left\{ \left(x,y\right)\equiv\left(x,y\right)+\delta\right\}} }$
from the quotient of (\ref{eq:parity2}) and removes the fiber $E_{b_{2}}$
from the quotient $\frac{dP_{9}}{\left\{ \left[s,t\right]\equiv\left[-s,t\right]\right\} }$,
the remaining open surfaces are isomorphic. Then $\frac{dP_{9}}{\left\{ \left[s,t\right]\equiv\left[-s,t\right]\right\} }$
corresponds to a flat $E_{8}$-bundle on $E_{b_{2}}/\left\{ \left(x,y\right)\equiv\left(x,y\right)+\delta\right\} )$
that pulls back to a flat $E_{8}$-bundle on $E_{b_{2}}$ that is invariant
under translation by $\delta$. Along $s=0$ the meromorphic two-form
\begin{equation}
d\ln s\wedge\frac{dx}{y}\label{eq:logdiff}
\end{equation}
is invariant under (\ref{eq:parity2}) so that it must be the one
that extends the the invariant two-form (\ref{eq:parity}). Therefore
its residue, $dx/y$ is also invariant under the $V_{3}$-involution
$\tilde{\beta}_{3}$. Said otherwise the quotient of the $\mathbb{Z}_{2}$-action
yields the order-$2$ logarthmic transform of each of the two components
\[
\left(dP_{9}^{\vee}/\mathbb{\mathbb{P}}_{a,\left[s_{a}^{2},t_{a}^{2}\right]}\right)\cup\left(dP_{9}^{\vee}/\mathbb{\mathbb{P}}_{b,\left[s_{b}^{2},t_{t}^{2}\right]}\right)
\]
yielding the $\mathbb{Z}_{2}$-action of 
\[
\frac{dx}{y}\mapsto\frac{dx}{y}
\]
on $V_{3}/B_{2}$. This action is necessary so that the Heterotic
quotient be a Calabi-Yau threefold. Simultaneously the action is consistent
with the $\mathbb{Z}_{2}$-action of 
\begin{equation}
\frac{dx}{y}\mapsto-\frac{dx}{y}\label{eq:-y}
\end{equation}
on $W_{4}/B_{3}$ that is necessary so that the $F$-theory quotient
be a Calabi-Yau fourfold.
\end{proof}
Lemma \ref{lem:i)-On-the F} is somewhat remarkable in its implications
for the $F$-theory dual. Since the involution $\beta_{2}$ on $B_{2}$
has only finite fixpoint set, it acts with eigenvalue $\left(+1\right)$
on the canonical bundle of $B_{2}$ . So Lemma \ref{lem:i)-On-the F}
says that for an $F$-theory dual with orbifold $\mathbb{Z}_{2}$
fundamental group\footnote{This will be useful for Wilson-line symmetry breaking.},
$\tilde{\beta}_{4}$ must act on the Weierstrass form of fibers of
the $F$-theory dual by 
\[
\left(x,y\right)\mapsto\left(x,-y\right).
\]
Only in that way does the $\mathbb{Z}_{2}$-quotient become a Calabi-Yau
fourfold.

\section{Retaining $E_{8}$-symmetry \label{sec:Retaining--symmetry3}}

As we have just shown, thanks to the nature of the logarithmic transform
above the fixpoints of $\beta_{2}$, the quotient
\[
\frac{W_{4,0}}{\tilde{\beta}_{4,0}}
\]
by the involution $\tilde{\beta}_{4,0}$ (induced on $W_{4,0}$ by
the involution $\tilde{\beta}_{4}$ on $W_{4}$) retains the structure
of the union of two $dP_{9}$'s (without multiple fibers). Therefore
quotienting the Heterotic model $V_{3}/B_{2}$ by $\tilde{\beta}_{3}/\beta_{2}$
does not break $E_{8}$-symmetry.

However (\ref{eq:inv2}) seems to imply that the involution $\tilde{\beta}_{4}/\beta_{3}$
does break $E_{8}$-symmetry on any crepant resolution $\tilde{W}_{4}/B_{3}$
of $W_{4}/B_{3}$. For example, in the idealized `limit' example where
the fibers of $W_{4}/B_{3}$ over points of $S_{\mathrm{GUT}}$ have
$E_{8}$-singular fibers
\[
y^{2}=x^{3}+a_{0}z^{5}
\]
the involution (\ref{eq:inv2}) sends each $E_{8}$-root as represented
by the exceptional fibers of the crepant resolution to its negative.
So the question becomes ``How can one endow the $F$-theory model
with an involution that leaves $E_{8}$ untouched but interchanges
the $E_{8}$-roots with their negatives?''

We propose that the answer lies with the interchange, over each point
$b_{2}\in B_{2}$, of the two complex structures 
\[
\left.\left(F_{a}^{\mathbb{C}}\oplus F_{b}^{\mathbb{C}}\right)\right|_{b_{2}}
\]
on the elliptic curve $E_{b_{2}}$ associated to the same flat real
$E_{8}$-bundles 
\[
\left.\left(F_{a}\oplus F_{b}\right)\right|_{b_{2}}.
\]
In support of this proposal we cite the following Lemma.
\begin{lem}
\label{lem:For-each-root}For each root $\rho$ of the compact real
form $G_{\mathbb{R}}$ of a simple algebraic group $G_{\mathbb{C}}$,
the involution $\iota$ exchanges the root space $\mathfrak{l}_{\rho}\subseteq g_{\mathbb{C}}$
with the root space $\mathfrak{l}_{-\rho}\subseteq g_{\mathbb{C}}$
and so acts as minus the identity on $\left[\mathfrak{l}_{\rho},\mathfrak{l}_{-\rho}\right]$
.
\end{lem}

\begin{proof}
For any pair of a root and its negative, consider the associated immersions
\begin{equation}
\begin{array}{ccc}
SU\left(2\right) & \rightarrow & G_{\mathbb{R}}\\
\downarrow &  & \downarrow\\
SL\left(2;\mathbb{C}\right) & \rightarrow & G_{\mathbb{C}},
\end{array}.\label{eq:equi}
\end{equation}
It will suffice to show the assertion for the realization of the compact
real form $SU\left(2\right)$ as the unit quaternions, its Lie algebra
as the real vector space corresponding to the imaginary quaternions
and the Lie algebra $\mathfrak{sl}\left(2;\mathbb{C}\right)$ of the
complex algebraic group $SL\left(2;\mathbb{C}\right)$. 

(\textit{For physicists only}) That is, it will suffice to show the
assertion for the roots of the compact real form $SU(2)$ with real
Lie algebra the trace-zero hermitian $2\times2$ matrices with basis
\begin{equation}
T_{3}:=\left(\begin{array}{cc}
1 & 0\\
0 & -1
\end{array}\right),-T_{2}:=\left(\begin{array}{cc}
0 & i\\
-i & 0
\end{array}\right),T_{1}:=\left(\begin{array}{cc}
0 & 1\\
1 & 0
\end{array}\right).\label{eq:real}
\end{equation}
considered as the real subspace of the complex Lie algebra $\mathfrak{sl}\left(2,\mathbb{C}\right)$
of trace-zero $2\times2$ matrices with basis given by adding respective
imaginary parts
\[
\left(\begin{array}{cc}
i & 0\\
0 & -i
\end{array}\right),\left(\begin{array}{cc}
0 & -1\\
1 & 0
\end{array}\right),\left(\begin{array}{cc}
0 & i\\
i & 0
\end{array}\right).
\]
With respect to the Cartan subalgebra generated by $T_{3}$ the root
spaces are given by the eigenvectors
\[
T_{1}\text{\textpm}i\text{·}T_{2}
\]
with real eigenvalues. These are exchanged by the action of Hermitian
conjugation.

(\textit{For mathematicians only}) That is, it will suffice to show
the assertion for the roots of the compact real form $SU(2)$ with
real Lie algebra the trace-zero skew-hermitian $2\times2$ matrices
with basis
\[
\mathbf{\mathbf{i}:=}\left(\begin{array}{cc}
i & 0\\
0 & -i
\end{array}\right),\mathbf{j}:=\left(\begin{array}{cc}
0 & 1\\
-1 & 0
\end{array}\right),\mathbf{k}:=\left(\begin{array}{cc}
0 & i\\
i & 0
\end{array}\right)
\]
considered as the real subspace of the complex Lie algebra $\mathfrak{sl}\left(2,\mathbb{C}\right)$
of trace-zero $2\times2$ matrices with basis given by adding respective
imaginary parts
\[
\left(\begin{array}{cc}
-1 & 0\\
0 & 1
\end{array}\right),\left(\begin{array}{cc}
0 & i\\
-i & 0
\end{array}\right),\left(\begin{array}{cc}
0 & -1\\
-1 & 0
\end{array}\right).
\]
With respect to the Cartan subalgebra generated by $\left(\begin{array}{cc}
i & 0\\
0 & -i
\end{array}\right)$ the root spaces are given by the eigenvectors 
\[
\left(\begin{array}{cc}
0 & 1\\
-1 & 0
\end{array}\right)\text{\textpm}\left(\begin{array}{cc}
0 & -1\\
-1 & 0
\end{array}\right)=\mathbf{j}\pm i\text{·}\mathbf{k}.
\]
These are then exchanged by the action of complex conjugation, and
the roots are purely imaginary and so go to minus themselves under
complex conjugation. 
\end{proof}
Said otherwise, we retain $E_{8}$-symmetry under the $\mathbb{Z}_{2}$-action
on the $F$- theory side by incorporating the involution $\iota$
into the $\mathbb{Z}_{2}$-action. In this way, we can admit the action
(\ref{eq:inv4}) that forces the reversal of choice of Weyl chamber
without breaking the symmetry with respect to the real group $E_{8}$
or with respect to the real group $SU\left(5\right)_{gauge}$.

Our claim is therefore that, in order to construct the $F$-theory
dual of a Heterotic theory in which $E_{8}$-symmetry is preserved
under a $\mathbb{Z}_{2}$-action, the quotient $F$-theory model can
only be endowed with initial $E_{8}$-symmetry if the $\mathbb{Z}_{2}$-action
incorporates the reversal of the choice of Weyl chamber. The reason
is that the choice of Weyl chamber is used to identify a system of
positive simple roots with exceptional components of the crepant resolution
$\tilde{W}_{4}/B_{3}$ of the $F$-theory model $W_{4}/B_{3}$. Otherwise
at the outset the $\mathbb{Z}_{2}$-action will simultaneously break
$E_{8}$-symmetry on the $F$-theory dual while maintaining $E_{8}$-symmetry
on the Heterotic model.

\section{Breaking $E_{8}$-symmetry to $SU\left(5\right)$ in $F$-theory\label{sec:Device-for-tracking4}}

Tracking the symmetry-breaking in $F$-theory and the Heterotic dual
begins by breaking symmetry of $E_{8}$ to that of the first factor of
a maximal subgroup 
\begin{equation}
\frac{SU\left(5\right)_{gauge}\times SU\left(5\right)_{Higgs}}{\mathbb{Z}_{5}}\hookrightarrow E_{8}.\label{eq:gpincl}
\end{equation}
The inclusion (\ref{eq:gpincl}) of rank-$8$ real compact semi-simple
Lie groups yields an identification of maximal abelian subalgebras
\begin{equation}
\mathfrak{h}_{SU\left(5\right)_{gauge}}\times\mathfrak{h}_{SU\left(5\right)_{Higgs}}\rightarrow\mathfrak{h}_{E_{8}},\label{eq:incl-1}
\end{equation}
an inclusion of Weyl groups
\begin{equation}
W\left(SU\left(5\right)_{gauge}\right)\times W\left(SU\left(5\right)_{Higgs}\right)\hookrightarrow W\left(E_{8}\right),\label{eq:incl2}
\end{equation}
and a morphism
\begin{equation}
\mathfrak{h}_{E_{8}}^{\ast}\rightarrow\mathfrak{h}_{SU\left(5\right)_{gauge}}^{\ast}\times\mathfrak{h}_{SU\left(5\right)_{Higgs}}^{\ast}\label{eq:map3}
\end{equation}
of roots and of the respective rings of Casimir polynomials. 

The associated symmetry-breaking is effected in terms of the `Tate
form' 
\begin{equation}
wy^{2}=x^{3}+a_{5}xyw+a_{4}zx^{2}w+a_{3}z^{2}yw^{2}+a_{2}z^{3}xw^{2}+a_{0}z^{5}w^{3}\label{eq:Tate}
\end{equation}
of the defining equation for the $F$-theory model $W_{4}/B_{3}$.
(\ref{eq:Tate}) defines $W_{4}$ as a hypersurface in a $\mathbb{P}^{2}$-bundle
\begin{equation}
P:=\mathbb{P}\left(\mathcal{O}_{B_{3}}\oplus\mathcal{O}_{B_{3}}\left(2N\right)\oplus\mathcal{O}_{B_{3}}\left(3N\right)\right)\label{eq:P-1}
\end{equation}
with homogeneous fiber coordinates $\left[w,x,y\right]$ over the
base $B_{3}$. $B_{3}$ is a Fano manifold that we will assume to
have very ample anti-canonical linear system whose generic divisor
we denote by $N$ and the the Calabi-Yau hypersurface $W_{4}\subseteq P$
is completely determined by the choice of
\[
z,a_{0},a_{2},a_{3},a_{4},a_{5},\frac{y}{x}=:t\in H^{0}\left(K_{B_{3}}^{-1}\right).
\]

By (\ref{eq:GUT}) we will require that
\[
z\circ\beta_{3}=-z
\]
so that by (\ref{eq:inv2})
\[
a_{j}\circ\beta_{3}=-a_{j}
\]
 for all $j$ and
\[
t\circ\beta_{3}=-t.
\]
Thus
\begin{equation}
z,a_{0},a_{2},a_{3},a_{4},a_{5},\frac{y}{x}=:t\in H^{0}\left(K_{B_{3}}^{-1}\right)^{\left[-1\right]},\label{eq:divcl-1}
\end{equation}
the $\left(-1\right)$-eigenspace with respect to the involution $\beta_{3}$
on $B_{3}$.\footnote{For purposes of avoiding vector-like exotics in the $F$-theory quotient,
we will always assume that 
\[
a_{0}=-\sum_{j=2}^{5}a_{j}.
\]
} We also initially assume that the $a_{2},a_{3},a_{4},a_{5}$ are chosen generically
in $H^{0}\left(K_{B_{3}}^{-1}\right)^{\left[-1\right]}$ . In particular
the map
\[
\psi_{3}=\left(a_{2},a_{3},a_{4},a_{5}\right):B_{3}\rightarrow\mathbb{P}^{3}
\]
is a finite morphism with the defining equation for $S_{\mathrm{GUT}}$
given by a (smooth) generic hyperplane section
\[
z=\sum_{j=2}^{5}\kappa_{j}a_{j}.
\]

\subsection{Tracking roots via rational double point surface singularities}
Our device for tracking the behavior of roots begins by rewriting
the equation of the $\mathrm{GUT}$-surface $S_{\mathrm{GUT}}$ as
\begin{equation}
z=a_{0}\text{·}\sum_{j=2}^{5}\kappa_{j}c_{j}\label{eq:zee-1}
\end{equation}
where $c_{j}=a_{j}/a_{0}$. Letting
\begin{equation}
\begin{array}{c}
B'_{3}:=B_{3}-\left\{ a_{0}=0\right\} \\
W_{4}':=W_{4}\times_{B_{3}}B'_{3}
\end{array}\label{eq:primedef}
\end{equation}
we divide (\ref{eq:Tate}) by $a_{0}^{6}$ and rescale by
\[
\begin{array}{c}
\frac{x}{a_{0}^{2}}\mapsto x\\
\frac{y}{a_{0}^{3}}\mapsto y\\
\frac{z}{a_{0}}\mapsto z
\end{array}
\]
to obtain
\begin{equation}
wy^{2}=x^{3}+c{}_{5}xyw+c_{4}zx^{2}w+c_{3}z^{2}yw^{2}+c_{2}z^{3}xw^{2}+z^{5}w^{3}\label{eq:bridge}
\end{equation}
with all entries invariant under the involution $\beta_{3}$ restricted
to $B_{3}'$. In particular $y$ now goes to $y$ under the $\mathbb{Z}_{2}$-action,
reflecting the fact that the Weyl chamber is no longer reversed when tracking the roots and $SU\left(5\right)$-symmetry
is preserved! It is only when wrapping a Wilson line on the non-contractible
loop on the $\mathbb{Z}_{2}$-quotient that symmetry is broken to
that of the Standard Model {[}MSSM{]}.

However to make this last equation compatible with the crepant resolution
of $W_{4}'/B_{3}'$, as in \cite{Donagi} where one has to interpret
the $c_{j}$ as Casimir polynomials giving the mapping 
\[
\left(c_{2},c_{3},c_{4},c_{5}\right):\mathfrak{h}_{SU\left(5\right)_{Higgs}}\rightarrow\frac{\mathfrak{h}_{SU\left(5\right)_{Higgs}}}{W\left(SU\left(5\right)\right)}=\mathbb{C}^{4},
\]
in order to equivariantly resolve the family (\ref{eq:bridge}), one
has to interpret the $c_{j}$ as Casimir polynomials giving the mapping 

\[
\left(c_{2},c_{3},c_{4},c_{5}\right):\mathfrak{h}_{SU\left(5\right)_{gauge}}\rightarrow\frac{\mathfrak{h}_{SU\left(5\right)_{gauge}}}{W\left(SU\left(5\right)\right)}=\mathbb{C}^{4}.
\]
Then in order to preserve $SU\left(5\right)$-symmetry on the quotient
of the $\mathbb{Z}_{2}$-action $\beta_{3}$, we must

1) replace each function $c_{j}$ on $\mathfrak{h}_{SU\left(5\right)_{gauge}}$
with the composed function
\[
c_{j}\circ\left(-I_{4}\right)
\]
where $-I_{4}$ takes each root to minus itself,

and 

2) send $y$ to $-y$ reflecting the action of $-I_{8}$ on $\mathfrak{h}_{E_{8}}$.

\noindent This will be explained in more detail in what follows.

By setting $w=1$, we will make
\begin{equation}
y^{2}=x^{3}+c{}_{5}xy+c_{4}zx^{2}+c_{3}z^{2}y+c_{2}z^{3}x+z^{5}\label{eq:link-1-1}
\end{equation}
a weighted homogeneous deformation of weight $30$ of the $E_{8}$
rational double point singularity
\[
y^{2}=x^{3}+z^{5}
\]
and simultaneously make (\ref{eq:link-1-1}) invariant with respect
to the involution induced by
\begin{equation}
\begin{array}{ccc}
\mathfrak{h}_{SU\left(5\right)_{gauge}}^{\mathbb{C}}\times\mathbb{C}^{3} &  & \mathfrak{h}_{SU\left(5\right)_{gauge}}^{\mathbb{C}}\times\mathbb{C}^{3}\\
\left(h,\left(x,y,z\right)\right) & \mapsto & \left(-h,\left(x,-y,z\right)\right)
\end{array}\label{eq:upmap}
\end{equation}
(that takes $c_{j}$ to $\left(-1\right)^{j}c_{j}$). 

\subsection{Deformation of the $E_{8}$ rational double point surface singularity}
To understand the implications of this last assertion, we begin by
choosing a nilpotent subregular element $X\in\mathfrak{e}_{8}$ whose
commutator contains $\mathfrak{h}_{E_{8}}^{\mathbb{C}}=\mathfrak{h}_{E_{8}^{\mathbb{C}}}$
and write elements of the the Lie subalgebra $\ker\left(ad\left(X\right)\right)$
as
\[
\left(\left(u,v\right),h,h'\right)\in\left(\mathbb{C}^{2}\times\mathfrak{h}_{SU\left(5\right)_{gauge}}\times\mathfrak{h}_{SU\left(5\right)_{Higgs}}\right)
\]
and restrict to the subspace
\[
\left(\left(u,v\right),h\right)\in\left(\mathbb{C}^{2}\times\mathfrak{h}_{SU\left(5\right)_{gauge}}\times\left\{ 0\right\} \right).
\]
To fit the product decomposition we must choose
\[
X=X_{gauge}^{sr}+X_{Higgs}^{r}
\]
where $X_{gauge}^{sr}\in\mathfrak{sl}\left(5;\mathbb{C}\right)_{gauge}$
is subregular and $X_{Higgs}^{r}\in\mathfrak{sl}\left(5;\mathbb{C}\right)_{Higgs}$
is regular and remark that their sum must be chosen to act faithfully
on the four non-principal summands of the adjoint action of $\mathfrak{sl}\left(5;\mathbb{C}\right)_{gauge}\times\mathfrak{sl}\left(5;\mathbb{C}\right)_{Higgs}$
on $\mathfrak{e}_{8}$.

The Jacobson-Morozov theorem \cite{Bourbaki} states that any nilpotent
element of $\mathfrak{sl}\left(5;\mathbb{C}\right)$, in particular
our subregular $X_{gauge}^{sr}$ in the commutant of our fixed $\mathfrak{h}_{SU\left(5\right)_{gauge}}^{\mathbb{C}}\subseteq\mathfrak{sl}\left(5;\mathbb{C}\right)$,
completes to an $\mathfrak{sl}\left(2;\mathbb{C}\right)$-triple $\bigl(X^{sr},Y^{sr},H^{sr}=\linebreak\left[X^{sr},Y^{sr}\right]\bigr)$.
Via (\ref{eq:incl-1}) that triple imbeds as a subalgebra of $\mathfrak{e}_{8}^{\mathbb{C}}$. 

We consider the following table of homogeneous forms\smallskip{}

\begin{center}
\begin{tabular}{|c|c|}
\hline 
Entry  & Weight=degree\tabularnewline
\hline 
\hline 
$x=f_{10}\left(u,v\right)$  & $10$\tabularnewline
\hline 
$y=f_{15}\left(u,v\right)$  & $15$\tabularnewline
\hline 
$z=f_{6}\left(u,v\right)$  & $6$\tabularnewline
\hline 
$c_{j}\left(h\right)$  & $j$\tabularnewline
\hline 
\end{tabular}\label{eq:table-1-1} \smallskip{}
\par\end{center}

\noindent on $\mathbb{C}^{2}\times\mathfrak{h}_{SU\left(5\right)_{gauge}}^{\mathbb{C}}$
with the property that the $f_{k}\left(u,v\right)$ are the generators
of the invariant polynomials of the finite subgroup of $SU\left(2\right)$
lying in $E_{8}$ via (\ref{eq:gpincl})  and the $c_{j}$ are the Casimirs that define
the mapping
\begin{equation}
\begin{array}{ccc}
\mathbb{C}^{2}\times\mathfrak{h}_{SU\left(5\right)_{gauge}}^{\mathbb{C}} & \rightarrow & \mathbb{C}^{3}\times\frac{\mathfrak{h}_{SU\left(5\right)_{gauge}}^{\mathbb{C}}}{W\left(SU\left(5\right)\right)}\\
\left(\left(u,v\right),h\right) & \mapsto & \left(\left(x,y,z\right),\left(c_{2},c_{3},c_{4},c_{5}\right)\right).
\end{array}\label{eq:singmap}
\end{equation}
As we will explain in more detail in the next Section, this gives
(\ref{eq:link-1-1}) the structure of a deformation 
\[
\mathcal{V}_{Tate}/\frac{\mathfrak{h}_{SU\left(5\right)_{gauge}}^{\mathbb{C}}}{W\left(SU\left(5\right)\right)}.
\]
of weighted homogeneous polynomials of weight $30$ of the $E_{8}$
rational double point surface singularity has image given by the equation
(\ref{eq:link-1-1}). Then the involution
\[
\left(\left(u,v\right),h\right)\mapsto\left(\left(-u,-v\right),-h\right)
\]
is equivariant with the involution induced by (\ref{eq:upmap}). 

Under the inclusion
\begin{equation}
\mathfrak{sl}\left(5;\mathbb{C}\right)_{gauge}\times\mathfrak{sl}\left(5;\mathbb{C}\right)_{Higgs}\hookrightarrow\mathfrak{e}_{8}^{\mathbb{C}}\label{eq:goodincl}
\end{equation}
that identifies maximal tori and hence Cartan subalgebras, we therefore
have the induced map
\begin{equation}
\begin{array}{c}
\frac{\mathfrak{h}_{SU\left(5\right)_{gauge}}^{\mathbb{C}}}{W\left(SU\left(5\right)_{gauge}\right)}\times\frac{\mathfrak{h}_{SU\left(5\right)_{Higgs}}^{\mathbb{C}}}{W\left(SU\left(5\right)_{Higgs}\right)}\rightarrow\frac{\mathfrak{h}_{E_{8}}^{\mathbb{C}}}{W\left(E_{8}\right)}\\
\left(\left(c_{2},c_{3},c_{4},c_{5}\right)_{gauge},\left(c_{2},c_{3},c_{4},c_{5}\right)_{Higgs}\right)\mapsto\left(a_{30},b_{24},b_{18},b_{12},c_{20},c_{14},c_{8},c_{2}\right)
\end{array}\label{eq:gauge-Higgs}
\end{equation}
where the coordinates of the $8$-dimensional vector space $\frac{\mathfrak{h}_{E_{8}}^{\mathbb{C}}}{W\left(E_{8}\right)}$
are indexed and weighted by the standard basis of Casimir polynomial algebra
of $E_{8}$.

Now the semi-universal deformation space of the rational double point
surface singularity
\begin{equation}
\left\{ y^{2}=x^{3}+z^{5}\right\} \subseteq\mathbb{C}\label{eq:E8-1}
\end{equation}
 is given by
\begin{align}
\label{eq:creptilda-1}
y^{2}&=x^{3}+z^{5}+a_{30}+\left(b_{24}z+b_{18}z^{2}+b_{12}z^{3}\right)\\
\notag &\quad +\left(c_{20}x+c_{14}xz+c_{8}xz^{2}+c_{2}xz^{3}\right)
\end{align}
where, in the first instance, the eight parameters $a_{j},b_{j},c_{j}$
are considered as free parameters of an eight-dimensional complex
vector space \cite{Brieskorn,Slodowy}.\linebreak  The semi-universal family
(\ref{eq:creptilda-1}) forms a hypersurface in $\mathbb{C}^{3}\times\frac{\mathfrak{h}_{E_{8}}^{\mathbb{C}}}{W\left(E_{8}\right)}$
\linebreak parametrized by the map 
\begin{equation}
\begin{array}{c}
\mathbb{C}^{2}\times\mathfrak{h}_{E_{8}}^{\mathbb{C}}\rightarrow\mathcal{V}_{8}\subseteq\mathbb{C}^{3}\times\frac{\mathfrak{h}_{E_{8}}^{\mathbb{C}}}{W\left(E_{8}\right)}\\
\left(\left(u,v\right),h\right)\mapsto\left(\left(x,y,z\right),\left(c_{2},c_{8},c_{14},c_{20},b_{12},b_{18},b_{24},a_{30}\right)\right).
\end{array}\label{eq:param}
\end{equation}
where one considers $\left(c_{2},c_{8},c_{14},c_{20},b_{12},b_{18},b_{24},a_{30}\right)$
as the standard generators of the Casimir polynomial algebra of $E_{8}$.
Again assigning\smallskip{}
 
\begin{center}
\begin{tabular}{|c|c|}
\hline 
Entry  & Weight\tabularnewline
\hline 
\hline 
$x=f_{10}$  & $10$\tabularnewline
\hline 
$y=f_{15}$  & $15$\tabularnewline
\hline 
$z=f_{6}$  & $6$\tabularnewline
\hline 
$a_{j},b_{j},c_{j}$  & $j$\tabularnewline
\hline 
\end{tabular}\label{eq:table-1-2} 
\par\end{center}

\smallskip{}
\noindent 
(\ref{eq:creptilda-1}) becomes a weighted homogeneous polynomial
of weight $30$ on 
\begin{equation}
\mathbb{C}^{2}\times\mathfrak{h}_{E_{8}^{\mathbb{C}}}.\label{eq:weightpar-1}
\end{equation}
We denote this semi-universal deformation as 
\begin{equation}
\mathcal{V}_{E_{8}}/\frac{\mathfrak{h}_{E_{8}}^{\mathbb{C}}}{W\left(E_{8}\right)}.\label{eq:V8}
\end{equation}
The fact that (\ref{eq:link-1-1}) is a weighted homogeneous deformation
of weight $30$ of the $E_{8}$ rational double point singularity
implies by semi-universality that it is induced by pullback 
\[
\begin{array}{ccc}
\mathcal{V}_{Tate}\times_{\frac{\mathfrak{h}_{SU\left(5\right)_{gauge}}^{\mathbb{C}}}{W\left(SU\left(5\right)_{gauge}\right)}}\mathfrak{h}_{SU\left(5\right)_{gauge}}^{\mathbb{C}} & \rightarrow & \mathcal{V}_{E_{8}}\\
\downarrow &  & \downarrow^{\pi}\\
\mathfrak{h}_{SU\left(5\right)_{gauge}}^{\mathbb{C}} & \rightarrow & \frac{\mathfrak{h}_{E_{8}}^{\mathbb{C}}}{W\left(E_{8}\right)}
\end{array}
\]
via (\ref{eq:gauge-Higgs}) from (\ref{eq:creptilda-1}). 

We next examine the semi-universal deformation (\ref{eq:V8}) in some
detail.

\section[Equivariant crepant resolution for the $E_{8}$ rational\\\mbox{} double point]{Equivariant crepant resolution for the $E_{8}$ rational double point\label{sec:Equivariant-crepant-resolution5}}

We again begin with the $E_{8}$-singularity (\ref{eq:E8-1}) whose
minimal resolution has the property that the intersection matrix of
its exceptional curves can be equated with the $E_{8}$-Dynkin diagram. (\ref{eq:E8-1})
is a quotient singularity via the forms 
\begin{equation}
\begin{array}{c}
x=f_{10}\left(u,v\right)\\
y=f_{15}\left(u,v\right)\\
z=f_{6}\left(u,v\right)
\end{array}\label{eq:generate-1-1-1}
\end{equation}
where $f_{j}\left(u,v\right)$ is a homogeneous form of degree $j$
in the complex $\left(u,v\right)$-plane. The exceptional curves themselves
are matched with simple roots corresponding to a choice of positive
Weyl chamber. The involution $\iota$ interchanges that choice of
positive Weyl chamber with its negative.

Letting $\varepsilon$ denote a primitive fifth root of unity, the
forms (\ref{eq:generate-1-1-1})are a minimal set of generators of
the sub-ring of the polynomial ring $\mathbb{C}\left[u,v\right]$
made up of the polynomials that are invariant under the action of
the the binary icosahedral group, that is, the finite subgroup $B\subseteq SU\left(2\right)$
of order $120$ generated by 
\[
\left(\begin{array}{cc}
\varepsilon^{3} & 0\\
0 & \varepsilon^{2}
\end{array}\right),\,\frac{1}{\sqrt{5}}\left(\begin{array}{cc}
-\varepsilon+\varepsilon^{4} & \varepsilon^{2}-\varepsilon^{3}\\
\varepsilon^{2}-\varepsilon^{3} & \varepsilon-\varepsilon^{4}
\end{array}\right).
\]
In fact 
\[
\begin{array}{ccc}
\mathbb{C}^{2} & \rightarrow & \mathbb{C}^{2}\\
\left(u,v\right) & \mapsto & \left(x,z\right)=\left(f_{10},f_{6}\right)
\end{array}
\]
has general fiber of cardinality $120$ on which $B$ acts transitively.
By degree, $f_{15}$ does not lie in the polynomial ring $\mathbb{C}\left[f_{10},f_{6}\right]$
however satisfies the second-degree integral equation 
\begin{equation}
wy^{2}=x^{3}+a_{0}z^{5}\label{eq:E8-1-1-1-1}
\end{equation}
and the degree of the mapping 
\begin{equation}
\begin{array}{ccc}
\mathbb{C}^{2} & \rightarrow & \mathbb{C}^{3}\\
\left(u,v\right) & \mapsto & \left(x,y,z\right)=\left(f_{10},f_{15},f_{6}\right)
\end{array}\label{eq:quotsg-1-1-1}
\end{equation}
is $240$. This is equivalent to the fact that 
\begin{equation}
f_{15}\left(-u,-v\right)=-f_{15}\left(u,v\right).\label{eq:flop-1-1}
\end{equation}
Thus the symmetry $\left(u,v\right)\mapsto\left(-u,-v\right)$ commutes
with the symmetry\linebreak $\left(x,y,z\right)\mapsto\left(x,-y,z\right)$.

The breaking of $E_{8}$-symmetry is tracked by an unfolding of (\ref{eq:E8-1})
in the semi-universal deformation space 
\begin{align}
\label{eq:creptilda-1-1}
y^{2}&=x^{3}+z^{5}+a_{30}+\left(b_{24}z+b_{18}z^{2}+b_{12}z^{3}\right)\\
\notag &\quad +\left(c_{20}x+c_{14}xz+c_{8}xz^{2}+c_{2}xz^{3}\right)
\end{align}
where, in the first instance, the eight parameters $a_{j},b_{j},c_{j}$
are considered as free parameters of an eight-dimensional complex
vector space that we will denote as 
\[
U_{8}:=\frac{\mathfrak{h}_{E_{8}}^{\mathbb{C}}}{W\left(E_{8}\right)}.
\]

Since the roots of the various subgroups of $E_{8}$ to which the
$E_{8}$-symmetry is broken are represented by the exceptional curves
of the crepant resolution of a rational double-point singularity over
a point of $U_{8}$ in (\ref{eq:param}), we will not be able to `follow
the roots' without understanding the Brieskorn-Grothendieck equivariant
crepant resolution of semi-universal deformation of the $E_{8}$-rational
double-point singularity as given in \cite{Brieskorn} and \cite{Slodowy}.

Unfortunately the equivariant Brieskorn-Grothendieck resolution cannot
be a resolution over $U_{8}$. Rather one considers $U_{8}$ as the
quotient 
\[
\left(c_{2},c_{8},c_{14},c_{20},b_{12},b_{18},b_{24},a_{30}\right):\mathfrak{h}_{E_{8}^{\mathbb{C}}}\rightarrow\frac{\mathfrak{h}_{E_{8}^{\mathbb{C}}}}{W\left(E_{8}^{\mathbb{C}}\right)}=U_{8}
\]
by considering the eight parameters $a_{j},b_{j},c_{j}$ in (\ref{eq:param})
as a standard basis of the $E_{8}^{\mathbb{C}}$ Casimir polynomials,
then assigning weights\smallskip{}
 
\begin{center}
\begin{tabular}{|c|c|}
\hline 
Entry  & Weight\tabularnewline
\hline 
\hline 
$x=f_{10}$  & $10$\tabularnewline
\hline 
$y=f_{15}$  & $15$\tabularnewline
\hline 
$z\equiv z=f_{6}$  & $6$\tabularnewline
\hline 
$a_{j},b_{j},c_{j}$  & $j$\tabularnewline
\hline 
\end{tabular}\label{eq:table-1} 
\par\end{center}

\smallskip{}
\noindent 
so that (\ref{eq:creptilda-1-1}) becomes a weighted homogeneous polynomial
of weight $30$ on 
\begin{equation}
\mathbb{C}^{2}\times\mathfrak{h}_{E_{8}^{\mathbb{C}}}.\label{eq:weightpar}
\end{equation}
From (\ref{eq:param}) we then have
\[
\begin{array}{ccc}
\mathbb{C}^{2}\times\mathfrak{h}_{E_{8}^{\mathbb{C}}} & \rightarrow & \mathcal{V}_{8}\times_{U_{8}}\mathfrak{h}_{E_{8}^{\mathbb{C}}}\\
\downarrow &  & \downarrow\\
\mathbb{C}^{2}\times U_{8} & \rightarrow & \mathcal{V}_{8}\subseteq\mathbb{C}^{3}\times U_{8}
\end{array}
\]

Somewhat miraculously, the equivariant crepant resolution $\tilde{\mathcal{V}}_{8}$
over $\mathfrak{h}_{E_{8}^{\mathbb{C}}}$ is canonically given by
subvarieties of the incidence variety 
\[
\mathcal{I}_{E_{8}}:=\left\{ \left(x,B\right):x\in B\right\} \subseteq E_{8}^{\mathbb{C}}\times\left\{ B\leq E_{8}^{\mathbb{C}}:B\,a\,Borel\,subgroup\right\} .
\]
To understand this, we begin with the regular elements of $E_{8}^{\mathbb{\mathbb{C}}}$,
that is those lying in only a finite number of Borel subgroups. Each
component of the commutant of a regular element is a maximal torus.
Choosing a maximal torus $\mathfrak{T}_{8}^{\mathbb{C}}$ for $E_{8}^{\mathbb{C}}$,
the product 
\[
\mathfrak{h}_{E_{8}^{\mathbb{\mathbb{C}}}}\times\mathbb{A}_{\left(u,v\right)}
\]
imbeds in the Lie algebra $\mathfrak{e}_{8}^{\mathbb{C}}$ as the
Lie algebra of the commutant subgroup $C\left(x_{sr}\right)$ of a
so-called subregular element $x_{sr}\in\mathfrak{T}_{8}^{\mathbb{C}}$
, that is, one whose commutant in $E_{8}^{\mathbb{C}}$ contains $\mathfrak{T}_{8}^{\mathbb{C}}$
as a codimension-two subgroup. The Lie $\mathbb{C}\text{·}H+\mathbb{A}_{\left(u,v\right)}$
should be considered as the Lie algebra of
\[
\frac{C\left(x_{sr}\right)}{\mathfrak{T}_{8}^{\mathbb{C}}}
\]
and, as such, one of the $\mathfrak{su}\left(2\right)\otimes\mathbb{C}$
Lie algebras in (\ref{eq:equi}), the complexification of the real
subalgebra generated by the two real matrices in (\ref{eq:real}).
Since $\iota$ acts as multiplication by $\left(-1\right)$ on roots,
it also takes a Borel subalgebra containing $\mathfrak{h}_{E_{8}^{\mathbb{\mathbb{C}}}}$
to its opposite Borel subalgebra. So if we let $u$ be the coordinate
for $\mathfrak{l}_{\varrho}$ and let $v$ be the coordinate for $\mathfrak{l}_{-\varrho}$
then via (\ref{eq:diagram}) $\iota$ induces the involution
\[
\left(h;u,v\right)\mapsto\left(-h;-u,-v\right)
\]
on $\mathfrak{h}_{E_{8}^{\mathbb{\mathbb{C}}}}\times\mathbb{A}_{\left(u,v\right)}$.

For the maximal torus $\mathfrak{T}_{8}^{\mathbb{C}}=\exp\left(\mathfrak{h}_{8}^{\mathbb{C}}\right)$
of $E_{8}^{\mathbb{C}}$ we next identify neighborhoods of the identity
under the exponential map 
\begin{equation}
\begin{array}{ccccccc}
\mathfrak{h}_{E_{8}^{\mathbb{\mathbb{C}}}}\times\mathbb{A}_{\left(u,v\right)} & \hookrightarrow & \mathfrak{e}_{8}^{\mathbb{C}} &  & C\left(x_{sr}\right) & \hookrightarrow & E_{8}^{\mathbb{C}}\\
\downarrow &  & \downarrow & \overset{\exp}{\Longrightarrow} & \downarrow &  & \downarrow\\
\mathfrak{h}_{E_{8}^{\mathbb{\mathbb{C}}}} & \rightarrow & \mathfrak{h}_{8}^{\mathbb{C}}/W\left(E_{8}\right) &  & \mathfrak{T}_{8}^{\mathbb{C}} & \rightarrow & \mathfrak{T}_{8}^{\mathbb{C}}/W\left(E_{8}\right)
\end{array}\label{eq:diagram}
\end{equation}
where the complex analytic map 
\[
E_{8}^{\mathbb{C}}\rightarrow\mathfrak{T}_{8}^{\mathbb{C}}/W\left(E_{8}\right)
\]
assigns to the conjugacy class of $x\in E_{8}^{\mathbb{C}}$ the well-defined
element $x_{s}\in\linebreak\mathfrak{T}_{8}^{\mathbb{C}}/W\left(E_{8}\right)$
of its Jordan decomposition $x=x_{s}x_{u}$ into commuting semi-simple
and unipotent factors.

In a small neighborhood of the identity the set of subregular elements
in $\mathfrak{T}_{8}^{\mathbb{C}}$ corresponds exactly under (\ref{eq:diagram})
to the set of the singular points of the versal deformation (\ref{eq:creptilda-1-1})
and the Brieskorn-Grothendieck equivariant crepant resolution 
\begin{equation}
\begin{array}{ccc}
\tilde{\mathcal{V}}_{8} & \rightarrow & \mathcal{V}_{8}\\
\downarrow &  & \downarrow\\
\mathfrak{h}_{E_{8}^{\mathbb{\mathbb{C}}}} & \rightarrow & U_{8}
\end{array}\label{eq:BG-1-1}
\end{equation}
has exception fibers over sub-regular element $x_{sr}$ given by subspaces
\[
\left\{ x_{sr}\right\} \times_{E_{8}}\mathcal{I}_{E_{8}}
\]
consisting of those Borels $B$ that contain $x_{sr}$. (See \cite{Brieskorn}
and \cite{Slodowy}.)
\begin{lem}
The action of the complex conjugate involution $\iota$ reverses a
choice of positive Weyl chamber of $E_{8}$ with its negative and
therefore reverse the choice of Weyl chamber with respect to which
the Brieskorn-Grothendieck equivariant crepant resolution is defined.
This reversal induces the involution (\ref{eq:inv2}) on on the semi-universal
deformation (\ref{eq:param}) of the $E_{8}$ rational double point
surface singularity. 
\end{lem}

\begin{proof}
The Brieskorn-Grothendieck equivariant crepant resolution is built
entirely inside the product of

1) the commutator of a sub-regular element of the complex algebraic
group $E_{8}^{\mathbb{C}}$

and

2) the set of Borel subgroups of $E_{8}^{\mathbb{C}}$.

Since all of the $E_{8}$-Casimirs are of even degree, $\left(u,v\right)\mapsto\left(-u,-v\right)$
commutes with the symmetry $\left(x,y,z\right)\mapsto\left(x,-y,z\right)$
in (\ref{eq:param}). But this last symmetry only commutes with the
Brieskorn-Grothendieck equivariant crepant resolution if it is induced
by a symmetry of (\ref{eq:weightpar}). Now the functions in Table
\ref{eq:table-1} are functions in the variables $\left(h;u,v\right)$
where the complex eight-tuple $h$ is the parameter for a neighborhood
of the origin in the Lie algebra of $E_{8}^{\mathbb{C}}$. The only
function of odd weight is
\[
y=f_{15}\left(h;u,v\right).
\]
Therefore the involution 
\[
\left(h;u,v\right)\mapsto\left(-h;-u,-v\right)
\]
induced by $\iota$ commutes with the projection to $U_{8}$ and sends
$\left(x,y,z\right)\mapsto\left(x,-y,z\right)$. 
\end{proof}
So the Brieskorn-Grothendieck equivariant crepant resolution is actually
a pair of crepant resolutions $\dot{\mathcal{V}}_{8}$ and $\ddot{\mathcal{V}}_{8}$
of the pullback 
\[
\mathcal{V}_{8}\times_{U_{8}}\mathfrak{h}_{E_{8}^{\mathbb{\mathbb{C}}}}
\]
of the family (\ref{eq:creptilda-1-1}) to a family over the Cartan
subalgebra $\mathfrak{h}_{E_{8}^{\mathbb{\mathbb{C}}}}$. The two
correspond to whether the resolution was grounded in a given choice
of positive Weyl chamber or its negative. The two resolutions  are
related by a real analytic isomorphism over $\mathcal{V}_{8}\times_{U_{8}}\mathfrak{h}_{E_{8}^{\mathbb{\mathbb{C}}}}$
induced by $\iota\in Gal\left(\mathbb{C}/\mathbb{R}\right)$. That
is, we have the commutative diagram 
\[
\begin{array}{ccc}
\dot{\mathcal{V}}_{8} & \overset{\iota}{\longleftrightarrow} & \ddot{\mathcal{V}}_{8}\\
\downarrow &  & \downarrow\\
\mathcal{V}_{8} & \overset{\left(x,y,z\right)\mapsto\left(x,-y,z\right)}{\longleftrightarrow} & \mathcal{V}_{8}.
\end{array}
\]
This diagram is then incorporated into a commutative diagram 
\begin{equation}
\begin{array}{ccc}
\dot{\mathcal{V}}_{8}\times_{\left(\mathcal{V}_{8}\times_{U_{8}}\mathfrak{h}_{E_{8}^{\mathbb{\mathbb{C}}}}\right)}\ddot{\mathcal{V}}_{8} & \rightarrow & \mathcal{V}_{8}\\
\downarrow &  & \downarrow\\
\mathfrak{h}_{E_{8}^{\mathbb{\mathbb{C}}}} & \rightarrow & U_{8}=\frac{\mathfrak{h}_{E_{8}^{\mathbb{\mathbb{C}}}}}{W\left(E_{8}^{\mathbb{C}}\right)}
\end{array}\label{eq:comm1-1}
\end{equation}
for which the left-hand vertical map is smooth on each factor of the
fibered product and the top horizontal map factors through crepant
resolutions $\dot{\mathcal{V}}_{8}\rightarrow\left(\mathcal{V}_{8}\times_{U_{8}}\mathfrak{h}_{E_{8}^{\mathbb{\mathbb{C}}}}\right)$
and $\ddot{\mathcal{V}}_{8}\rightarrow\left(\mathcal{V}_{8}\times_{U_{8}}\mathfrak{h}_{E_{8}^{\mathbb{\mathbb{C}}}}\right)$
respectively.

These two resolutions have the the following four properties:

1) Each is determined by a choice of positive Weyl chamber in $\mathfrak{h}_{8}^{\mathbb{C}}$.

2) The mappings 
\begin{equation}
\dot{\mathcal{V}}_{8},\,\mathcal{\ddot{V}}_{8}/\mathfrak{h}_{8}^{\mathbb{C}}\rightarrow\left(\mathcal{V}_{8}\times_{U_{8}}\mathfrak{h}_{8}^{\mathbb{C}}\right)/\mathfrak{h}_{8}^{\mathbb{C}}\label{eq:eqres-1}
\end{equation}
are isomorphisms except over the singular locus of $\mathcal{V}_{8}\times_{U_{8}}\mathfrak{h}_{8}^{\mathbb{C}}$.

3) Over a singular point $\left(z=0,x=0,y=0;h\right)$ the fiber is
the Dynkin curve for the minimal resolution of the rational double-point
singularity corresponding to $h\in\mathfrak{h}_{8}^{\mathbb{C}}$.

4) Since all weights in Table (\ref{eq:table-1}) except that of $y$
are even and $\iota$ acts as minus the identity on the vector space
$\mathfrak{h}_{E_{8}^{\mathbb{\mathbb{C}}}}\times\mathbb{A}_{\left(u,v\right)}$,
we conclude that $\iota$ acts fiberwise on the semi-universal family
(\ref{eq:creptilda-1-1}) as 
\begin{equation}
\left(u,v,x,y,z\right)\mapsto\left(-u,-v,x,-y,z\right).\label{eq:fibinv-1}
\end{equation}

We therefore are able to conclude the following. 
\begin{thm}
\label{thm:-symmetry-breaking-in-Heterotic} Via the Brieskorn-Grothendieck
equivariant crepant resolution of the $E_{8}$ rational double-point
(\ref{eq:E8-1}), the action of the generator $\iota\in Gal\left(\mathbb{C}/\mathbb{R}\right)$
corresponds to the action 
\begin{equation}
\left(x,y\right)\mapsto\left(x,-y\right)\label{eq:Weier}
\end{equation}
on the Weierstrass form on the fibers of the semi-universal deformation
of the $E_{8}$ rational double point. 
\end{thm}

\section[Tracking the equivariant resolution under\\\mbox{} symmetry-breaking]{Tracking the equivariant resolution under symmetry-breaking\label{sec:Tracking-the-equivariant6}}

\subsection{Immersing $W_{4}/B_{3}$ into the versal deformation of the $E_{8}$
rational double point surface singularity}

Next notice that the above mapping
\begin{equation}
\left(c_{2},c_{3},c_{4},c_{5}\right):B'_{3}\rightarrow\mathbb{C}^{4}=\frac{\mathfrak{h}_{SU\left(5\right)_{gauge}}^{\mathbb{C}}}{W\left(SU\left(5\right)\right)}\label{eq:c3-1-1}
\end{equation}
is such that, by (\ref{eq:link-1-1}) coupled with (\ref{eq:zee-1}),
it defines a commutative diagram 
\begin{equation}
\begin{array}{ccc}
W_{4}' & \rightarrow & \mathcal{V}_{Tate}\\
\downarrow &  & \downarrow\\
B'_{3} & \rightarrow & \mathbb{C}^{4}.
\end{array}\label{eq:keycom-1}
\end{equation}

We next return to the isomorphism
\[
\mathfrak{h}_{SU\left(5\right)_{gauge}}\times\mathfrak{h}_{SU\left(5\right)_{Higgs}}\rightarrow\mathfrak{h}_{E_{8}}
\]
given in (\ref{eq:incl-1}) inducing the epimorphism (\ref{eq:incl-1}).
We pull back the family $\mathcal{V}_{Tate}/\frac{\mathfrak{h}_{SU\left(5\right)_{gauge}}^{\mathbb{C}}}{W\left(SU\left(5\right)_{gauge}\right)}$
given by (\ref{eq:link-1-1}) under the map
\[
\mathfrak{h}_{SU\left(5\right)_{gauge}}^{\mathbb{C}}\times\frac{\mathfrak{h}_{SU\left(5\right)_{Higgs}}^{\mathbb{C}}}{W\left(SU\left(5\right)_{Higgs}\right)}\rightarrow\frac{\mathfrak{h}_{E_{8}^{\mathbb{\mathbb{C}}}}}{W\left(E_{8}^{\mathbb{C}}\right)}
\]
where it becomes a weighted homogeneous family of rational double
points of weight $30$. Therefore by the semi-universality of the
family (\ref{eq:creptilda-1}) we have the induced commutative diagram
\begin{equation}
\begin{array}{ccc}
\mathcal{V}_{Tate}\times_{\frac{\mathfrak{h}_{SU\left(5\right)_{gauge}}^{\mathbb{C}}}{W\left(SU\left(5\right)_{gauge}\right)}}\mathfrak{h}_{SU\left(5\right)_{gauge}}^{\mathbb{C}} & \rightarrow & \mathcal{V}_{E_{8}}\\
\downarrow &  & \downarrow^{\pi}\\
\mathfrak{h}_{SU\left(5\right)_{gauge}}^{\mathbb{C}} & \rightarrow & \frac{\mathfrak{h}_{E_{8}}^{\mathbb{C}}}{W\left(E_{8}\right)}
\end{array}\label{eq:keycom-2}
\end{equation}
and so by §8 of \cite{Slodowy} equivariant crepant resolutions
\begin{equation}
\begin{array}{ccc}
\left\{ \mathcal{\dot{V}}_{Tate},\,\ddot{\mathcal{V}}_{Tate}\right\} \times_{\frac{\mathfrak{h}_{SU\left(5\right)_{gauge}}^{\mathbb{C}}}{W\left(SU\left(5\right)_{gauge}\right)}}\mathfrak{h}_{SU\left(5\right)_{gauge}}^{\mathbb{C}} & \rightarrow & \mathcal{\dot{V}}_{E_{8}},\,\ddot{\mathcal{V}}_{E_{8}}\\
\downarrow &  & \downarrow^{\tilde{\pi}}\\
\mathfrak{h}_{SU\left(5\right)_{gauge}}^{\mathbb{C}} & \rightarrow & \mathfrak{h}_{E_{8}}^{\mathbb{C}}.
\end{array}\label{eq:res}
\end{equation}
Now referring to the composition 
\begin{equation}
\begin{array}{ccccc}
W_{4}' & \hookrightarrow & \mathcal{V}_{Tate} & \rightarrow & \mathcal{V}_{E_{8}}\\
\downarrow &  & \downarrow &  & \downarrow^{\pi}\\
B_{3}' & \hookrightarrow & \mathbb{C}^{4}=\frac{\mathfrak{h}_{SU\left(5\right)_{gauge}}^{\mathbb{C}}}{W\left(SU\left(5\right)_{gauge}\right)} & \rightarrow\mathbb{C}^{8}= & \frac{\mathfrak{h}_{E_{8}}^{\mathbb{C}}}{W\left(E_{8}\right)}
\end{array}\label{eq:comdia}
\end{equation}
of (\ref{eq:keycom-1}) and (\ref{eq:keycom-2}), (\ref{eq:res})
lets us track roots over the central column of (\ref{eq:comdia}).
It remains to track those roots as given by (\ref{eq:res}) to the exceptional curves of a crepant resolution $\tilde{W}_{4}/B_{3}$
of $W_{4}/B_{3}$, at least over a general point $p\in S_{\mathrm{GUT}}$. 

Before proceeding to accomplish this last task, notice that $\iota=-I_{8}$
acts as
\begin{align*}
&\left(\left(a_{30},b_{24},b_{18},b_{12},c_{20},c_{14},c_{8},c_{2}\right),\left(x,y,z\right)\right)\\ &\qquad \mapsto\left(\left(a_{30},b_{24},b_{18},b_{12},c_{20},c_{14},c_{8},c_{2}\right),\left(x,-y,z\right)\right)
\end{align*}
on the right-hand vertical map in (\ref{eq:comdia}), as
\[
\left(\left(c_{2},c_{3},c_{4},c_{5}\right),\left(x,y,z\right)\right)\mapsto\left(\left(c_{2},-c_{3},c_{4},-c_{5}\right),\left(x,-y,z\right)\right)
\]
on the central vertical map in (\ref{eq:comdia}), and as 
\[
\left(b_{3},\left(x,y,z\right)\right)\mapsto\left(\beta_{3}\left(b_{3}\right),\left(x,-y,-z\right)\right)
\]
on the left-hand vertical map in (\ref{eq:comdia}).\footnote{The lack of a sign change of $y$ between the central vertical map
of (\ref{eq:comdia}) and the left-hand vertical map under the action
of $-I_{8}$ reflects the fact that the action $y\mapsto-y$ of the
involution $\beta_{3}$ on $B_{3}'$ will incorporate the reversal
of positive roots with their negatives induced by the complex conjugate
involution $-I_{8}$. That `flop' interchanges the equivarant crepant
resolutions $\mathcal{\dot{V}}_{E_{8}}/\mathfrak{h}_{E_{8}}^{\mathbb{C}}$
and $\ddot{\mathcal{V}}_{E_{8}}/\mathfrak{h}_{E_{8}}^{\mathbb{C}}$.
It is only in this way that (\ref{eq:link-1-1}) becomes invariant
under the action of $\beta_{3}$. } 
\begin{thm}
(\ref{eq:keycom-1}) allows us to track a crepant resolution of $W_{4}'/B_{3}'$
back to the crepant resolution of the $E_{8}$ rational double point
singularity over a general point $p\in S_{\mathrm{GUT}}\cap B_{3}'$.
\end{thm}

\begin{proof}
For general $p\in S_{\mathrm{GUT}}$, we define a holomorphic disk
$D_{p}\subseteq B_{3}$ meeting $S_{\mathrm{GUT}}$ transversely at
$p$ and form
\[
D_{p}\times_{B_{3}}W_{4},
\]
a smooth open elliptic surface whose resolution has an $I_{5}$-fiber
over $p$ by the Kodaira classification. Now $D_{p}$ determines a
normal vector $\nu_{p}$ to $S_{\mathrm{GUT}}\subseteq B_{3}$ at
$p$ that by the map from the left-hand column to the central column
of (\ref{eq:comdia}) lifts to a non-zero nilpotent subregular element
$X_{gauge}^{sr}=\tilde{\nu}_{p}$ in the nilpotent cone of the complex
Lie algebra $\mathfrak{g}_{SU\left(5\right)}^{\mathbb{C}}$ and so
into the nilpotent cone of the complex Lie algebra $\mathfrak{g}_{E_{8}}^{\mathbb{C}}$.
We complete $X_{gauge}^{sr}=\tilde{\nu}_{p}$ to an 
\[
\mathfrak{sl}_{2}^{\mathbb{C}}\,triple\subseteq\mathfrak{sl}_{5}^{\mathbb{C}}\subseteq\mathfrak{g}_{E_{8}}^{\mathbb{C}}
\]
via the Jacobson-Morosov theorem. Now the $I_{5}$-fiber over $p$
in the crepant resolution $\tilde{W}_{4}'/B_{3}'$ of $W_{4}'/B_{3}'$
implies that the decomposition of $\mathfrak{g}_{SU\left(5\right)_{gauge}}^{\mathbb{C}}$
as an $\mathfrak{sl}_{2}^{\mathbb{C}}$-module induced by the triple
must have a simple summand decomposition that coincides with a simple
decomposition associated with the $I_{5}$-fiber of $\mathcal{\dot{V}}_{Tate}/\mathfrak{h}_{SU\left(5\right)_{gauge}}^{\mathbb{C}}$,
respectively $\ddot{\mathcal{V}}_{Tate}/\mathfrak{h}_{SU\left(5\right)_{gauge}}^{\mathbb{C}}$,
and so of $\mathcal{\dot{V}}_{E_{8}}/\mathfrak{h}_{E_{8}}^{\mathbb{C}}$,
respectively $\ddot{\mathcal{V}}_{E_{8}}/\mathfrak{h}_{E_{8}}^{\mathbb{C}}$,
over a lifting of the image of $p$ in $\frac{\mathfrak{h}_{E_{8}}^{\mathbb{C}}}{W\left(E_{8}\right)}$.
Thus the fiber over $p$ of the crepant resolution $\tilde{W}_{4}'/B_{3}'$
of $W_{4}'/B_{3}'$ is induced by the fiber of the equivariant crepant
resolution $\mathcal{\tilde{V}}_{E_{8}}/\mathfrak{h}_{E_{8}}^{\mathbb{C}}$
over a lifting of the image of $p$ in $\frac{\mathfrak{h}_{E_{8}}^{\mathbb{C}}}{W\left(E_{8}\right)}$.
\end{proof}

\section{Crepant resolution conjecture\label{sec:Crepant-resolution-conjecture7}}

The diagram (\ref{eq:comdia}) and (\ref{eq:BG-1-1}) induce a commutative
diagram
\[
\begin{array}{ccccc}
W_{4}'\times_{B_{3}'}\mathfrak{h}_{SU\left(5\right)_{gauge}}^{\mathbb{C}} & \overset{\tilde{\vartheta}}{\dashrightarrow} & \tilde{\mathcal{V}}_{E_{8}} & \rightarrow & \mathcal{V}_{E_{8}}\\
\downarrow &  & \downarrow &  & \downarrow^{\pi}\\
B_{3}'\times_{\frac{\mathfrak{h}_{SU\left(5\right)_{gauge}}^{\mathbb{C}}}{W\left(SU\left(5\right)_{gauge}\right)}}\mathfrak{h}_{SU\left(5\right)_{gauge}}^{\mathbb{C}} & \rightarrow & \mathfrak{h}_{E_{8}}^{\mathbb{C}} & \rightarrow & \frac{\mathfrak{h}_{E_{8}}^{\mathbb{C}}}{W\left(E_{8}\right)}
\end{array}
\]
\begin{conjecture}
The closure of the graph of the rational map $\tilde{\vartheta}$
in the above diagram is a crepant resolution of $W_{4}'\times_{B_{3}'}\mathfrak{h}_{SU\left(5\right)_{gauge}}^{\mathbb{C}}$.
\end{conjecture}

If true, this conjecture would allow is to track a crepant resolution
of $W_{4}'/B_{3}'$ back to the crepant resolution of the $E_{8}$
rational double point singularity over \textit{every} $p\in S_{\mathrm{GUT}}\cap B_{3}'$,
for example over the points of the matter and Higgs curves.

\section{Conclusion}

In this paper we have confronted a problem proposed but not fully
resolved in \cite{Beasley}, \cite{Donagi-2} and \cite{Marsano}.
Our solution rests on the introduction of the complex conjugation
operator into the $\mathbb{Z}_{2}$-action on $W_{4}/B_{3}$ to produce
a Calabi-Yau quotient on which we still retain $SU\left(5\right)_{gauge}$-symmetry.
This last assertion is proved by tracing the exceptional fibers of
a crepant resolution $\tilde{W}_{4}/B_{3}$ of $W_{4}/B_{3}$ back
to the $E_{8}$-roots from which they evolved using the Brieskorn-Grothendieck
equivariant resolution of the semi-universal deformation of the $E_{8}$
rational double-point surface singularity. 

One is still left with the task of explicitly constructing the $B_{3}$,
the resolution $\tilde{W}_{4}/B_{3}$, and checking that the $\mathbb{Z}_{2}$-quotients
have the phenomenologically correct invariant. Our strategy will be
to first construct $B_{3}$ canonically from the geometry of $A_{4}$-roots
space in such a way that it is both symmetric with respect to the
action of complex conjugation and also has the desired numerical invariants.
That done, the construction of $\tilde{W}_{4}/B_{3}$ and verification
that it too has the correct numerical invariants will be relatively
straightforward. The authors carried this program out in two related
papers, ``F-theory over a Fano threefold built from $A_{4}$-roots" [arXiv: hep-th/1912.06902] 
and ``Heterotic-F-theory Duality with Wilson line Symmetry-breaking"  [arXiv: hep-th/1908.01913].

\section*{Acknowledgements}

The authors would like to thank Dave Morrison, Tony Pantev and Sakura
Schäfer-Nameki for their guidance and many helpful conversations over
several years. However the authors themselves take sole responsibility
for any errors or omissions in this paper. S.R. acknowledges partial
support from Department of Energy grant DE-SC0011726.

\address{Department of Mathematics, The Ohio State University\\
Columbus, OH 43210-1174, USA\\
\email{clemens.43@osu.edu}\\
\email{raby.1@osu.edu}}

\end{document}